\documentclass[conference,twocolumn,letterpaper]{IEEEtran}
\usepackage{graphicx}%
\usepackage[T1]{fontenc}
\usepackage{lmodern} %
\usepackage{amssymb,stmaryrd,amsmath,amsfonts}
\usepackage{amsmath,amsthm,amssymb,cases}
\usepackage{algorithm}
\usepackage{tikz}
\usetikzlibrary{positioning}
\usepackage{algpseudocode}
\usepackage{amsmath,cases}
\usepackage{bbm}
\usepackage{amsthm}
\usepackage{booktabs}
\usepackage{color}
\usepackage{amssymb}
\usepackage{bm}
\usepackage{thmtools}

\theoremstyle{definition}%

\newtheorem{yosei}{Requirement}
\newtheorem{df}{Definition}
\newtheorem{teiri}{Theorem}

\newtheorem{example}{Example}
\usepackage{rotating}
\usepackage{mathtools}
\mathtoolsset{showonlyrefs}

\usepackage{bbm}
\usepackage{mathrsfs}

\makeatletter
\newcommand\RedeclareMathOperator{%
  \@ifstar{\def\rmo@s{m}\rmo@redeclare}{\def\rmo@s{o}\rmo@redeclare}%
}
\newcommand\rmo@redeclare[2]{%
  \begingroup \escapechar\m@ne\xdef\@gtempa{{\string#1}}\endgroup
  \expandafter\@ifundefined\@gtempa
     {\@latex@error{\noexpand#1undefined}\@ehc}%
     \relax
  \expandafter\rmo@declmathop\rmo@s{#1}{#2}}
\newcommand\rmo@declmathop[3]{%
  \DeclareRobustCommand{#2}{\qopname\newmcodes@#1{#3}}%
}
\@onlypreamble\RedeclareMathOperator
\makeatother

\definecolor{grey}{rgb}{0.7, 0.75, 0.71}

\def\dark-red#1{\textcolor[rgb]{0.7,0.0,0.0}{#1}}

\definecolor{amber}{rgb}{1.0, 0.75, 0.0}

\def\...{\dotsc}

\def\intT2{\int_{-T/2}^{T/2}}

\def\sumi1n{\sum_{i=1}^{n}}
\def\sumi1N{\sum_{i=1}^{N}}
\def\sumi0N--{\sum_{i=0}^{N-1}}

\def\ccc{\cdots}

\def\=def{\overset{\text{\small def}}{=}}

\DeclarePairedDelimiterX{\inp}[2]{\langle}{\rangle}{#1, #2}

\def\Fb{\mathbb{F}}

\def\Zb{\mathbb{Z}}

\def\Ac{\mathcal{A}}

\def\Fc{\mathcal{F}}

\def\Qc{\mathcal{Q}}

\def\<{\langle}
\def\>{\rangle}

 \def\mat4#1#2#3#4{
\begin{pmatrix}
 #1&\ccc&#2\\
 \vdots&&\vdots\\
 #3&\ccc&#4
\end{pmatrix}}

\def\0sf{\mathsf{0}}
\def\1sf{\mathsf{1}}

\def\0BS{\boldsymbol{0}}
\def\1BS{\boldsymbol{1}}

\def\0B{\mathbf{0}}
\def\1B{\mathbf{1}}

\def\0H{\hat{0}}
\def\1H{\hat{1}}

\def\HH{\hat{H}}

\def\+TT{\texttt{+}}
\def\-{\texttt{-}}
\def\+KB{|+\> \<+|}
\def\-KB{|-\> \<-|}

\def\q0{|0\>}

\def\0U{\underline{0}}
\def\1U{\underline{1}}

\def\fU{\underline{f}}
\def\gU{\underline{g}}

\def\vU{\underline{v}}

\def\0UH{\underline{\0H}}
\def\1UH{\underline{\1H}}

\RedeclareMathOperator{\Im}{Im}

\def\thepage{\arabic{page}}
\newcommand{\T}{\mathsf{T}}

\ifCLASSINFOpdf
\else
\fi
\begin{document}
\title{Quantum Error Correction with Girth-16 Non-Binary LDPC Codes via Affine Permutation Construction}
\author{
\IEEEauthorblockN{Kenta Kasai}
\IEEEauthorblockA{
Institution of Science Tokyo\\
Email: kenta@ict.eng.ist.ac.jp}
}
\maketitle

\begin{abstract}
We propose a method for constructing quantum error-correcting codes
based on non-binary low-density parity-check codes
with Tanner graph girth 16.
While conventional constructions using circulant permutation matrices
are limited to girth 12,
our method employs affine permutation matrices
and a randomized sequential selection procedure
to eliminate short cycles and achieve girth 16.

Numerical experiments show that the proposed codes
significantly reduce the number of low-weight codewords.
Joint belief propagation decoding over depolarizing channels reveals
that although a slight degradation appears in the waterfall region,
a substantial improvement is achieved in the error floor performance.

We also evaluated the minimum distance
and found that the proposed codes achieve a larger upper bound
compared to conventional constructions.
\end{abstract}
\IEEEpeerreviewmaketitle

\section{Introduction}

Recent advancements in quantum computing have enabled the construction of systems with tens of
reliable logical qubits from thousands of noisy physical qubits~\cite{preskill2018quantum}.

However, scaling quantum computations to thousands or more logical qubits remains a major challenge.
This underscores the critical need for highly efficient quantum error correction (QEC) techniques
capable of supporting such large-scale systems.

A promising approach to QEC, based on non-binary low-density parity-check (LDPC) codes, was proposed
in~\cite{6017122}.
This construction utilized circulant permutation matrices (CPMs), which, however, inherently limited
the Tanner graph girth to at most 12~\cite{1317123} and imposed strict constraints on submatrix
sizes \cite{4557323}.
These limitations significantly restricted design flexibility.

To overcome these challenges, a recent extension~\cite{komoto2024quantumerrorcorrectionnear}
replaced CPMs with general permutation matrices.
By applying joint belief propagation decoding, the resulting codes achieved performance close to the
hashing bound, with no observable error floor across a wide range of experiments.

In this paper, we push the boundary further by proposing a randomized construction of quantum LDPC
codes with girth~16, using affine permutation matrices (APMs).
Our method, detailed in Section~\ref{sec:proposed_construction}, enhances the girth while preserving
the desirable properties of previous constructions.

We validate our approach through numerical experiments conducted over depolarizing channels with
joint belief propagation decoding, as described in Section~\ref{sec:results}.
The proposed codes significantly reduce the number of low-weight codewords and improve error floor
performance, albeit with a slight degradation in the waterfall region compared to conventional
designs.
We also evaluate the minimum distance of the proposed codes.

\section{Permutation Matrices for LDPC Codes}

Permutation matrices play a central role in the construction of structured quantum LDPC codes.
In this section, we systematically define three important classes of permutation matrices:
general permutation matrices (PMs), APMs, and CPMs. 

\subsection{Permutation Matrices}
Let $P$ be a positive integer.
Define $\Zb_P = \{0, 1, \dots, P - 1\}$, and let $\mathcal{F}_P$ denote the set of all bijections
from $\Zb_P$ to itself.
The identity permutation is denoted by $\mathrm{id}$.
Each permutation $f \in \mathcal{F}_P$ is associated with a permutation matrix
$F \in \mathbb{F}_2^{P\times P}$ according to the rule:
\begin{align}
f(j) = i \quad \iff \quad F_{i,j} = 1.
\end{align}
We denote this correspondence by $f \sim F$, or simply write $f = F$ when the context is clear.

Throughout the paper, we will freely switch between viewing $f$ as a function and as its associated
matrix $F$.
The following properties hold for any $f, g \in \mathcal{F}_P$ and their corresponding matrices
$F, G$:
\[
f \circ g \sim FG, \quad F^\top \sim f^{-1}.
\]
\subsection{Affine and Circulant Permutation Matrices}
 For integers $a, b \in \Zb_P$, define an affine permutation $f: \Zb_P \to \Zb_P$ by
\[
f(j) = a j + b \mod P.
\]
It is known that $f \in \Fc$ if and only if $\gcd(a, P) = 1$~\cite{myung2006combining}.  
A permutation matrix arising from a permutation of this form is called an APM. We denote the set of all such matrices by $\Ac_P$. 
When $a = 1$, the affine permutation reduces to a simple shift: $f(j) = j + b \mod P$. The corresponding permutation matrices are called CPMs, and we denote their set by $\Qc_P$.

Let us illustrate a simple case of $P = 4$.
\begin{example}
 Let $a = 3$ and $b = 1$. Then
$ f(j) = 3j + 1,$ which yields:
$f(0) = 1,\quad f(1) = 0,\quad f(2) = 3,\quad f(3) = 2.$
The corresponding APM $F$ is
\begin{align}
F &= \small
\left(
\begin{array}{cccc}
0 & 1 & 0 & 0 \\
1 & 0 & 0 & 0 \\
0 & 0 & 0 & 1 \\
0 & 0 & 1 & 0
\end{array}
\right).
\end{align}
\end{example}

\subsection{Low-Density Parity-Check Matrices}
\label{045449_15Feb25}
Let $\HH \in \Fc_P^{J \times L}$ be a permutation array whose entries are taken from a set of permutations $\Fc_P$. 
This array can also be interpreted as a binary parity-check matrix in $\Fb_2^{JP \times LP}$.
LDPC codes defined by such parity-check matrices are conventionally referred to as protograph codes, APM-LDPC codes, and quasi cyclic (QC)-LDPC codes when the permutations are taken from $\Fc_P$, $\Ac_P$, and $\Qc_P$, respectively.
As an example, see the parity-check matrices $\HH_X$ and $\HH_Z$ in Example~\ref{021639_17Apr25}.
\subsection{Block Cycles}
We regard the parity-check matrix $\HH$ as a binary matrix consisting of $J \times L$ blocks.  
We consider a block path in $\HH$ that alternates horizontally and vertically across the matrix, starting and ending at the same block.  
Such a path is represented as
\[
f_{11} \to f_{12} \to f_{21} \to f_{22} \to \cdots \to f_{n1} \to f_{n2} \to f_{11}
\]
within the following submatrix $S$ of $\HH$:
\begin{align}
S= \left(
\begin{array}{llllll}
  f_{11}        & f_{12}        & *                & *                & *                & *                \\
  *             & f_{21}        & f_{22}           & *                & *                & *                \\
  *             & *             & \ddots           & \ddots           & *                & *                \\
  *             & *             & *                & *                & f_{n1}           & f_{n2}           \\
  f_{n2}        & *             & *                & \cdots           & *                & f_{n1}
\end{array}
\right).
\end{align}
Note that consecutive blocks must differ in both row and column indices. 
For instance, the column indices of $f_{11}$ and $f_{12}$ must be different, and likewise, the row indices of $f_{12}$ and $f_{21}$ must be different.

The sequence $(f_{11}, f_{12}, \dots, f_{n,1}, f_{n,2})$ consists of elements in $\Fc_P$.  
We define the \emph{composite function} associated with the block cycle as
\begin{align}
f_*(j) := \bigl(f_{n2}^{-1} f_{n1} \cdots f_{22}^{-1} f_{21} f_{12}^{-1} f_{11}\bigr)(j), \quad j \in \Zb_P,
\end{align}
where the composition operator $\circ$ is omitted for simplicity.  
Equivalently, the inverse function can be written as
\begin{align}
f_*^{-1}(j) = \bigl(f_{11}^{-1} f_{12} \cdots f_{n1}^{-1} f_{n2}\bigr)(j), \quad j \in \Zb_P.
\end{align}
We call the block cycle \emph{closed} if $f_*(j) = j$ for some $j \in \Zb_P$, and \emph{open} otherwise.  
In particular, if $f_*(j) = j$ for all $j \in \Zb_P$, that is, if $f_* = \mathrm{id}$, the block cycle is said to be \emph{totally closed}.
The existence of a closed block cycle implies the presence of a cycle in the corresponding Tanner graph~\cite{myung2006combining,yoshida2019linear}.  
Conversely, if no closed block cycle exists, then the Tanner graph contains no cycles.
We define the \emph{girth} of a given parity-check matrix as the length of the shortest closed block cycle contained in it.

When $J = 2$, each step in the block cycle alternates between the first and second row blocks, and hence its length is always a multiple of 4.  
In such cases, if we assume that the cycle starts from a block in the upper row and that the first move is in the horizontal direction, the block cycle can be uniquely specified by the sequence of column indices.

Furthermore, if each entry in the first row block is distinct—meaning that it uniquely determines the column in which it appears—then the cycle can be completely described by the sequence of permutation elements.

Let us consider the matrix $\hat{H}_X$ in Example~\ref{021639_17Apr25}.  
Since the entries in the upper row are distinct, each one determines a unique column index.  
Therefore, the following block sequence:
\begin{align}
&\overline{f_0}\to \overline{g_0}\to \underline{g_3}\to \underline{f_0}\to \overline{f_1}\to \overline{g_3}\to \underline{g_2}\to \underline{f_1}
\\&\to \overline{f_2}\to \overline{g_2}\to \underline{g_1}\to \underline{f_2}\to \overline{f_3}\to \overline{g_1}\to \underline{g_0}\to \underline{f_3}\to \overline{f_0}
\end{align}
(where elements from the upper row are denoted with overlines and those from the lower row with underlines)
can be compactly represented by the sequence:
\begin{align}
f_0\to g_0\to f_1\to g_3\to f_2\to g_2\to f_3\to g_1\to f_0.
\end{align}

From the following theorem, it follows that any QC-LDPC code whose parity-check matrix contains a $\Qc_P$-valued submatrix of size at least $2 \times 3$ has girth at most 12.
\begin{teiri}[\cite{1317123}]\label{144806_5Apr25}
Let $a, b, c, d, e, f \in \mathcal{F}_P$ be mutually commuting elements.  
In particular, this commutativity condition is automatically satisfied when $a, b, c, d, e, f \in \Qc_P$.
Consider the following $2 \times 3$ subarray of a larger permutation array $H$:
\begin{align}
\left( 
\begin{array}{ccc}
a & b & c \\
d & e & f
\end{array}
\right).
\end{align}
Then, the block cycle of length 12 formed by the sequence $a \to b \to c \to a \to b \to c \to a$
is a closed. Moreover, it forms a totally closed block cycle.  
We refer to such a block cycle as a \emph{$2 \times 3$ totally closed block cycle}.
\end{teiri}
\begin{proof}
Since all the elements commute, the following composition reduces to the identity function:
\begin{align}
f_*^{-1}= a^{-1}b  e^{-1} f  c^{-1} a  d^{-1} e  b^{-1} c  f^{-1} d   = \mathrm{id}.
\end{align}
\end{proof}

\subsection{Komoto--Kasai Construction}
In \cite{komoto2024quantumerrorcorrectionnear}, the construction method for Calderbank--Shor--Steane (CSS) codes based on non-binary LDPC codes, originally proposed for $\Qc_P$-valued arrays in \cite{6017122}, was generalized to support $\Fc_P$-valued arrays.
In this framework, we construct the parity-check matrices $H_X$ and $H_Z$ that define a CSS code from two sequences of permutations $\fU := (f_0, \ldots, f_{L/2-1})$ and $\gU := (g_0, \ldots, g_{L/2-1})$ in $\Fc_P^{L/2}$.

The following requirement on $\fU$ and $\gU$ is used to ensure that the resulting matrices $\HH_X$ and $\HH_Z$ are orthogonal \cite{komoto2024quantumerrorcorrectionnear}.
\begin{yosei}\label{065413_4Feb25}
Let $f_0, \ldots, f_{L/2-1},\ g_0, \ldots, g_{L/2-1} \in \Fc_P$.  
Denote $[n]=\{0,1,\ldots,n-1\}$. 
We require that the following commutativity condition holds:
\begin{align}
  f_{\ell-j} \, g_{k-\ell} = g_{k-\ell} \, f_{\ell-j} \quad \text{for all } \ell \in [L/2],\ j,k \in [J]. \label{162154_15Feb25}
\end{align}
 Here, the indices of $f$ and $g$ are understood modulo $L/2$. 
\end{yosei}
This commutativity condition guarantees that each pair of corresponding blocks in $\HH_X$ and $\HH_Z$ ensures the orthogonality.
\begin{df}
 Let $f_0, \dots, f_{L/2-1},\ g_0, \dots, g_{L/2-1} \in \mathcal{F}_P$ be permutations indexed cyclically, i.e., $f_k = f_{L/2 + k}$ and similarly for $g_k$. 
Define the following four $\Fc_P$-valued ${J\times L/2}$ matrix blocks as:
 \begin{align}
 (\hat{H}^{(\mathrm{L})}_X)_{j,l} &= f_{-j + l},& (\hat{H}^{(\mathrm{R})}_X)_{j,l} &= g_{-j + l}, \\
 (\hat{H}^{(\mathrm{L})}_Z)_{j,l} &= g^{-1}_{j - l}, &(\hat{H}^{(\mathrm{R})}_Z)_{j,l} &= f^{-1}_{j - l}.
 \end{align}
 Define the full matrices as:
 \begin{align}
 \hat{H}_X &= (\hat{H}^{(\mathrm{L})}_X \mid \hat{H}^{(\mathrm{R})}_X), \hat{H}_Z = (\hat{H}^{(\mathrm{L})}_Z \mid \hat{H}^{(\mathrm{R})}_Z).
 \end{align}
 Here, the indices $j, l$ are understood modulo $L/2$. 
\end{df}
The exchange of the $f$- and $g$-blocks between $\hat H_X$ and $\hat H_Z$ is essential for the orthogonality argument and agrees with Example~\ref{021639_17Apr25}.
This construction generalizes the original Hagiwara--Imai construction~\cite{4557323,6017122}, which imposed strong algebraic constraints, to a more flexible framework involving general PMs~\cite{komoto2024quantumerrorcorrectionnear}.
Under the condition in Requirement~\ref{065413_4Feb25}, the matrices $\HH_X, \HH_Z \in \Fb_2^{JP \times LP}$ are orthogonal\footnote{
$(\HH_X\HH_Z^\T)_{jk}=\sum_{i} F_{i-j} G_{k-i} + \sum_{i} G_{i-j} F_{k-i} 
{=} \sum_{i} F_{i-j} G_{k-i} + \sum_{i} F_{k-i} G_{i-j} 
= O.
$}, i.e., $\HH_X \HH_Z^\top = O$.
\begin{example}\label{021639_17Apr25}
Let $J = 2$ and $L = 8$.
Using the APMs or CPMs $f_0, \ldots, f_3$ and $g_0, \ldots, g_3$ specified in Table~\ref{015622_17Apr25}, 
we define the following $J \times L$ protograph-based permutation array:
\begin{align}
\hat{H}_X &=
\left[
\begin{array}{c@{\hspace{6mm}}c@{\hspace{6mm}}c@{\hspace{6mm}}c@{\hspace{6mm}}c@{\hspace{6mm}}c@{\hspace{6mm}}cc}
f_0 & f_1 & f_2 & f_3 & g_0 & g_1 & g_2 & g_3 \\
f_3 & f_0 & f_1 & f_2 & g_3 & g_0 & g_1 & g_2
\end{array}
\right], \\
\hat{H}_Z &=
\left[
\begin{array}{cccccccc}
g_0^{-1} & g_3^{-1} & g_2^{-1} & g_1^{-1} & f_0^{-1} & f_3^{-1} & f_2^{-1} & f_1^{-1} \\
g_1^{-1} & g_0^{-1} & g_3^{-1} & g_2^{-1} & f_1^{-1} & f_0^{-1} & f_3^{-1} & f_2^{-1}
\end{array}
\right].
\end{align}
The arrangement of blocks is carefully designed to satisfy the orthogonality condition $\HH_X \HH_Z^\top = 0$ and to achieve girth 16 and 12.
The construction of the proposed code is described in detail in the next section.
\end{example}

Next, we construct non-binary parity-check matrices $H_\Gamma, H_\Delta \in \mathbb{F}_q^{JP \times LP}$ with $q = 2^e$, and then generate binary parity-check matrices $H_X, H_Z \in \mathbb{F}_2^{eJP \times eLP}$ by replacing each non-zero element with the $e \times e$ companion matrix over $\mathbb{F}_2$ associated with the corresponding finite field element, or its transpose. For further details on this conversion procedure, we refer the reader to \cite{komoto2024quantumerrorcorrectionnear,6017122}.
The CSS code defined by the binary parity-check matrices $H_X$ and $H_Z$—that is, the codes $C_X$ and $C_Z$ they define—is the quantum code used in this study for quantum error correction.
\section{Proposed Construction Method of $C_X,C_Z$}\label{sec:proposed_construction}
When attempting to satisfy the orthogonality condition required by the CSS code construction (Requirement~\ref{065413_4Feb25}), 
totally closed block cycles of length $2L$ inevitably arise~\cite{komoto2025explicitconstructionclassicalquantum}.  
Therefore, in order to construct codes with girth at least 16, it is necessary to set $L \ge 8$.
In this paper, we focus on the case $L = 8$ as a concrete instance. 
The proposed method closely follows Komoto--Kasai construction~\cite{komoto2024quantumerrorcorrectionnear}, except for the design of the matrices $\HH_X$ and $\HH_Z$, which requires a specific modification.

As shown in Theorem~\ref{144806_5Apr25}, any parity-check matrix that includes a $\Qc_P$-valued $2 \times 3$ subarray necessarily has Tanner graph girth at most 12.  
Therefore, to overcome the inherent girth limitation of $\Qc_P$, we take a decisive step forward and harness the full flexibility of $\Ac_P$ to construct our matrices.
\begin{table}[htbp]\label{015542_17Apr25}
\caption{APMs and CPMs $(f_i, g_i)$ used in the proposed and conventional constructions \cite{komoto2025explicitconstructionclassicalquantum} for $P = 12600$.}
\label{015622_17Apr25}
\begin{align}
\renewcommand{\arraystretch}{1.2}
\begin{array}{c||c|c||c|c}
&\multicolumn{2}{c||}{\text{Prop. code  with girth 16}} & \multicolumn{2}{c}{\text{Conv. code with girth 12}} \\\hline
i & f_i(x) & g_i(x) & f_i(x) & g_i(x) \\\hline
0&3151X+7075&6301X+5178&X+4375&X+2833\\
1&9451X+6495&5041X+9360&X+11775&X+11168\\
2&7351X+1295&X+4584&X+7825&X+6792\\
3&10501X+3540&7561X+5784&X+11351&X+ 3961\\
\end{array}
\end{align}
\end{table}

We aim to construct a pair of sequences of affine permutations
\[
\fU = (f_0, f_1, \ldots, f_{L/2 - 1}), \quad \gU = (g_0, g_1, \ldots, g_{L/2 - 1}),
\]
that satisfy the following criteria:
\begin{itemize}
  \item[(a)] Each pair $(f_i, g_i)$ should satisfy the commutativity condition required in Requirement~\ref{065413_4Feb25}.
  \item[(b)] To avoid the formation of $2 \times 3$ totally closed block cycles (see Theorem~\ref{144806_5Apr25}), commutativity among the $f_i$'s or among the $g_i$'s must be avoided.  
  Specifically, the functions $f_i$ and $f_j$ ($i \ne j$), and likewise $g_i$ and $g_j$, should not commute.
  \item[(c)] The constructed matrices $\HH_X$ and $\HH_Z$ must be free of any closed block cycles of length up to 12 in order to ensure the desired girth property.
\end{itemize}
\textit{Note:} Although condition (a) may appear to be implied by (c), it plays a critical role in the construction process.  
In practice, if condition (a) is not explicitly enforced, it becomes extremely difficult to find sequences satisfying (c) via random sampling.

To construct such sequences $\fU$ and $\gU$, we use a sequential randomized algorithm that adds each function one at a time while checking the required constraints.  
Starting from an empty list, we iteratively select candidates for $f_i$ and $g_i$ by sampling from $\mathcal{A}_P$ and verifying that the resulting partial sequences do not violate conditions~(a), (b), or (c).  
If a candidate does not violate any constraints, it is accepted into the sequence; otherwise, a new candidate is drawn.  
This process is repeated until valid sequences $\fU$ and $\gU$ of length $L/2$ are constructed.  
The precise procedure is described in Algorithm~\ref{alg:fg_construction}.

\renewcommand\algorithmicindent{.5em}
\begin{algorithm}[thbp]
\caption{Sequential Construction of $\fU, \gU$}
\label{alg:fg_construction}
\begin{algorithmic}[1]
\State Initialize empty list $\mathcal{S} \gets [\ ]$
\For{$i = 0$ to $L/2 - 1$}
    \Repeat
        \State Randomly generate a candidate $f_i$ from $\mathcal{A}_P$
        \State Temporarily set $\mathcal{S}' \gets \mathcal{S} \cup \{f_i\}$
        \If{$\mathcal{S}'$ does not violate the criteria (a),(b), or (c)}
            \State Accept $f_i$: $\mathcal{S} \gets \mathcal{S}'$
            \State \textbf{break}
        \EndIf
    \Until{a valid $f_i$ is found}
    \Repeat
        \State Randomly generate a candidate $g_i$ from $\mathcal{A}_P$
        \State Temporarily set $\mathcal{S}' \gets \mathcal{S} \cup \{g_i\}$
        \If{$\mathcal{S}'$ does not violate the criteria (a),(b), or (c)}
            \State Accept $g_i$: $\mathcal{S} \gets \mathcal{S}'$
            \State \textbf{break}
        \EndIf
    \Until{a valid $g_i$ is found}
\EndFor
\State \textbf{return} $\mathcal{S}$
\end{algorithmic}
\end{algorithm}

\section{Numerical Results}\label{sec:results}

Using the algorithm described in the previous section,
we constructed sequences $\fU$ and $\gU$,
and the corresponding matrices $\HH_X$ and $\HH_Z$,
for parameters $J = 2$, $L = 8$, and $P = 12600$\footnote{
Although choosing $P$ as a power of two would be practical,
we observed that it made it difficult to satisfy both constraints (a) and (b).
We found that $P$ needs to contain several small prime factors to admit valid sequences,
and successfully constructed such sequences when $P$ was a multiple of $6300 = 2 \cdot 3^2 \cdot 5 \cdot 7$.
}.

Table~\ref{015542_17Apr25} lists the pairs $(f_i, g_i)$ of APMs used in the proposed method,
as well as those used in the conventional CPM-based method.

We set $e = 8$ and $q = 2^e$,
and used the primitive polynomial $x^8 + x^4 + x^3 + x^2 + 1$ over $\mathbb{F}_2$ to define the finite field $\mathbb{F}_q$.
The non-binary $\Fb_q$-valued parity-check matrices $H_\Gamma$ and $H_\Delta$ were constructed from $\HH_X$ and $\HH_Z$
by assigning nonzero elements in $\mathbb{F}_q$ according to the block structure.
Each nonzero element was replaced by the corresponding $e \times e$ binary companion matrix~\cite{komoto2025explicitconstructionclassicalquantum},
or its transpose, to obtain the binary matrices $H_X$ and $H_Z$.

We conducted numerical experiments using joint belief propagation decoding~\cite{komoto2025explicitconstructionclassicalquantum}
over depolarizing channels with error probability~$p$,
and marginal error probability $f_\mathrm{m} = \frac{2p}{3}$.
Decoding was considered successful when the noise pattern was correctly estimated;
otherwise, it was deemed a failure.
Degenerate errors were not considered.

\begin{figure}
  \centering
  \includegraphics[width=1.00\linewidth]{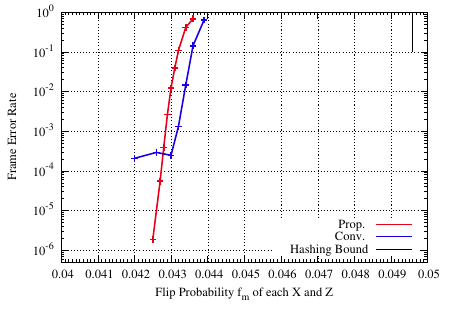}
  \caption{
  Frame error rate comparison between the proposed $[[806400, 403200, \le\! 14]]$ code and the conventional $[[806400, 403200,\allowbreak \le 9]]$ code~\cite{komoto2024quantumerrorcorrectionnear}.
  }
  \label{fig:comparison}
\end{figure}

The proposed method exhibits slightly inferior performance in the waterfall region compared to the conventional method.
This behavior is consistent with the well-known trade-off in LDPC code design:
algebraically constructed codes achieve better error floor performance but may perform worse in the waterfall region~\cite{lin2004structured}.
Conversely, randomly constructed LDPC codes generally perform well in the waterfall region but often suffer from higher error floors.

As discussed later,
the proposed codes achieve a larger upper bound on the minimum distance,
which we believe to be tight.
Thus, we expect the proposed codes to eventually outperform the conventional ones in the error floor region.

In our numerical experiments,
the conventional codes exhibited an observable error floor around a frame error rate of $10^{-4}$,
whereas the proposed codes did not exhibit any noticeable error floor down to at least $10^{-6}$.
\footnote{In the initial version of our arXiv submission, an undetected bug in the error floor detection process led us to incorrectly conclude that the conventional codes exhibited no error floor.
This issue has been identified and rectified in the present study.}
\subsection{Computing Upper Bound of Minimum Distance}
For the proposed code $(C_X, C_Z)$ constructed with girth~16, we first enumerated all shortest cycles.  
For example, consider a submatrix $S$ of $\Fb_q$-valued parity-check matrix $H_\Gamma$ corresponding to a length-16 cycle, such as the following:
\[
S=\left(
\begin{array}{cccccccc}
\gamma_{11}& \gamma_{12}&0            &0            &0            \\
0          & \gamma_{21}& \gamma_{22} &0            &0            \\
0          &0           & \ddots      & \ddots      &0            \\
0          &0           &0            & \gamma_{71} & \gamma_{72} \\
\gamma_{82}&0           &0            &0            & \gamma_{81}
\end{array}
\right).
\]
Let $\vU \in \mathbb{F}_q^8$ be a nonzero vector in the null space of $S$.  
By padding zeros at positions outside the columns of $S$, $\vU$ can be embedded into a length-$PL$ vector over $\mathbb{F}_q$, which belongs to $C_X$.  
The null space of $S$ has dimension either 0 or 1.  
In the former case, it contains only the zero vector; in the latter case, it contains $q - 1$ nonzero codewords.
In total, we found 7190 and 7256 such length-16 cycles whose corresponding submatrices have null space dimension 1 in $C_X$ and $C_Z$, respectively.
By computing the bitwise weights of such codewords that do not belong to $C_Z^\perp$, we obtained the weight distribution $A_X(w)$, as listed in Table~\ref{weight_dist}.  
 Similarly, we computed $A_Z(w)$ from the null spaces of submatrices of $H_\Delta$ corresponding to length-16 cycles.
For comparison, we also computed the weight distributions $A_X(w)$ and $A_Z(w)$ for the conventional codes $(C_X, C_Z)$ constructed with girth~12.
From the table, we obtain upper bounds on the minimum bitwise distances of 14 for the proposed code and 9 for the conventional code.  

From the table, we observe that the minimum bitwise weight among these codewords is 14.  
However, this upper bound on the minimum distance may not be exact, since it is possible that other codewords in $C_X$—arising from longer cycles or combinations of cycles in the Tanner graph—have bitwise weights smaller than those enumerated here, even though their weights in $\mathbb{F}_q$ are at least 9.  
Nevertheless, due to the use of a relatively large finite field size $q$, we believe that this upper bound is likely to be tight.
\begin{table}[thbp]
  \centering
\caption{Weight distribution of codewords formed from the shortest cycles. Here, $A_X(w)$ and $A_Z(w)$ denote the number of codewords of weight $w$ in $C_X \setminus C_Z^\perp$ and $C_Z \setminus C_X^\perp$, respectively.}
\label{weight_dist}
  \renewcommand{\arraystretch}{0.95}
\begin{tabular}{c|c|c||c|c|c}
\multicolumn{3}{c||}{\text{Prop. }}&\multicolumn{3}{c}{\text{Conv. }}  \\\hline
    \hline
    $w$ & $A_X(w)$ & $A_Z(w)$&w& $A_X(w)$ & $A_Z(w)$\\
    \hline
    14  & 0     & 2     &     9  & 0     & 7     \\
    15  & 7     & 4     &    10  & 3     & 10    \\
    16  & 19    & 16    &    11  & 23    & 36    \\
    17  & 64    & 74    &    12  & 115   & 117   \\
    $\vdots$&$\vdots$&$\vdots$&$\vdots$&$\vdots$&$\vdots$\\
     50  & 4     & 7      &    38  & 17    & 19    \\
     51  & 1     & 1      &    39  & 5     & 3     \\
     52  & 0     & 1      &    40  & 0     & 2     \\
 \hline\hline
 total & 255$\times$ 7192 & 255$\times$ 7256 & total  & 255$\times$ 3155	 & 255$\times$ 3063
  \end{tabular}
\end{table}

\section{Conclusion and Future Work}
\label{sec:conclusion}
We proposed a method for constructing quantum LDPC codes with Tanner graph girth~16
using affine permutation matrices and a randomized sequential selection process.
Unlike conventional constructions limited to girth~12, our method effectively eliminates
short cycles that degrade decoding performance.

Numerical experiments showed that the proposed codes reduce the number of low-weight codewords, leading to a substantial improvement
in the error floor region, although with slightly inferior performance in the waterfall region.
The resulting CSS codes achieve a higher upper bound on the minimum distance, which we conjecture to be tight.

Overall, our findings suggest that exploiting affine permutation matrices combined with randomized construction
is an effective strategy to mitigate short cycles and improve error floor performance in quantum LDPC codes.

Future work includes extending the construction to larger $L$, applying spatial coupling techniques~\cite{CSS_SC_ISIT},
and developing analytical methods to predict error floor behavior based on low-weight codewords
and small trapping sets.
It is also conceivable to extend the explicit construction method of~\cite{komoto2025explicitconstructionclassicalquantum}
to affine permutation matrices.
We also plan to apply the method proposed in~\cite{kasai2025efficient} to further enhance
error floor performance.

\section*{Acknowledgment}
This study was carried out using the TSUBAME4.0 supercomputer at Institute of Science Tokyo.
\bibliographystyle{IEEEtran}
\bibliography{IEEEabrv,ref} 
\end{document}